\documentclass[12pt,a4paper,twoside]{article}%
\usepackage{amsthm,amsmath,amssymb,amsfonts,pifont}
\usepackage[varg]{txfonts}
\usepackage{enumerate}
\usepackage{indentfirst}
\usepackage{graphicx}
\usepackage{subfigure}
\usepackage{tikz}
\usepackage[numbers,sort&compress]{natbib}
\usepackage[colorlinks,urlcolor=blue,linkcolor=blue,citecolor=blue]{hyperref}
\usepackage{hypernat}
\usepackage{geometry}
\usepackage{anysize}

\usetikzlibrary{arrows,shapes,chains,3d}
\marginsize{2.5cm}{2.5cm}{2.5cm}{2.5cm}
\newtheorem{theorem}{Theorem}[section]
\newtheorem{lemma}[theorem]{Lemma}

\theoremstyle{definition}
\newtheorem{remark}{Remark}[section]
\newtheorem{definition}{Definition}[section]

\numberwithin{equation}{section}
\numberwithin{figure}{section}
\allowdisplaybreaks
\hypersetup{CJKbookmarks={true}}

\begin{document}
	\title{\Large Signal Encryption Strategy Based on Domain Change of the 
	Fractional Fourier Transform}
	
	\author{Wei Chen$^{1,2}$, Zunwei Fu$^1$, Loukas Grafakos$^3$, Yue Wu$^1$}
	\date{\small\it
		$^1$ School of Mathematics and Statistics, Linyi University, Linyi 
		276000, China\\
		$^2$ College of Information Technology, The University of Suwon, 
		Hwaseong-si 18323, South Korea\\
		$^3$ Department of Mathematics, University of Missouri, Columbia MO 
		65211, USA
	}
	
	\maketitle
	
	\begin{abstract}
	This paper provides a double encryption algorithm that uses the lack 
	of invertibility of the fractional Fourier transform (FRFT) on $L^{1}$. One 
	encryption key is a function, which maps a ``good" $L^{2}$-signal to a 
	``bad" 	$L^{1}$-signal. The FRFT parameter which describes the rotation 
	associated 	with this operator on the time-frequency plane provides the 
	other encryption key. With the help of approximate identities, such as of 
	the Abel and Gauss means of the FRFT established in \cite{CFGW}, we recover 
	the encrypted signal on the FRFT domain. This design of an encryption 
	algorithm seems new even  when using the classical Fourier transform. 
	Finally, the feasibility of the new strategy is verified by simulation and 
	audio examples.		
	
	\end{abstract}

\section{Introduction}

\label{sect:intro}

In view of the rapid development of communication and multimedia technology,
the acquisition, transmission and processing of private data is of paramount
importance in our digital era. Alongside with the development of digital
technology, security concerns arise as a major undertaking. The implementation
of data encryption and related algorithm research are key pillars to solving
these problems. In the past decade, scholars have applied this research in
various fields, such as telemedicine, patient data secure filing, voice
communication, satellite transmission signal, and so on; on this we refer to
\cite{RanPh,chao,RegShEn, JntTranCor,jigsaw,iter,Ozak} and references therein.

In the last two decades, the fractional Fourier transform (FRFT) has been
widely favored in engineering, in view of its free parameters, its suitability
in dealing with rotations in the time-frequency plane, and its convenience in
algorithm adaptation; in fact, a number of algorithms have been established
based on the FRFT. Moreover, the FRFT has found applications in aspects of
research, such as artificial neural network, wavelet transform, time-frequency
analysis, time-varying filtering, complex transmission and so on (see, e.g.,
\cite{S2011,T2010,D2001,N2003,M2002,Y2003}). Moreover, it has also been used
in partial differential equations (cf., \cite{namias,Lohmann93}), quantum
mechanics (cf., \cite{namias,PhysRevLett72}), diffraction theory and optical
transmission (cf., \cite{Ozak}), optical system and optical signal processing
(cf., \cite{Bernardo94,Ozak,Liu97}), optical image processing (cf.,
\cite{Lohmann93,Liu97}), etc. Earlier theoretical aspects of the FRFT can be
found in \cite{namias,MK1987,Kerr1988,KF1988,Lohmann93,zayed}.

The FRFT can be interpreted as a form of the Fourier transform which
incorporates a rotation of the coordinate axis on which the original signal is
defined; this rotation is given counterclockwise about the origin in the
time-frequency plane. Theoretically, we introduce the FRFT as an operator on
$L^{1}(\mathbb{R})$ as follows: For $u\in L^{1}(\mathbb{R})$ and $\alpha
\in\mathbb{R}$, the fractional Fourier transform of order $\alpha$ of
$u$ is defined by
\begin{equation}
(\mathcal{F}_{\alpha}u)(x)=\left\{
\begin{array}
[c]{ll}%
\int_{-\infty}^{+\infty}K_{\alpha}(x,t)u(t)\,\mathrm{d}t, & \alpha\neq
n\pi,\quad n\in\mathbb{N},\\
u(x), & \alpha=2n\pi,\\
u(-x), & \alpha=(2n+1)\pi,
\end{array}
\right.  \label{eq:fa}%
\end{equation}
where
\[
K_{\alpha}(x,t)=A_{\alpha}\exp\left[  2\pi i\left(  \frac{t^{2}}{2}\cot
\alpha-xt\csc\alpha+\frac{x^{2}}{2}\cot\alpha\right)  \right]
\]
is the kernel of FRFT and
\begin{equation}
A_{\alpha}=\sqrt{1-i\cot\alpha}. \label{defAa}%
\end{equation}
It is obvious that when $\alpha=\pi/2$, the FRFT reduces to the ordinary
Fourier transform, that is, $\mathcal{F}_{\pi/2}=\mathcal{F}$. Recall that the
Fourier transform of $u$ defined as%
\begin{equation}
(\mathcal{F}u)(x)=\int_{-\infty}^{+\infty}u(t)e^{-2\pi ixt}\mathrm{d}t.
\label{eq:FT}%
\end{equation}
If the Fourier transform of a signal is another signal that lives on an
axis perpendicular to the original signal time axis, the $\alpha$th FRFT of a
signal lives on the counterclockwise rotation by the angle $\alpha$ of the
original signal time axis.

To the best of our knowledge, researchers have only applied the $L^{2}$ theory
of FRFT and did not consider an $L^{1}$ theory. As the FRFT of $L^{1}$-signal,
$\mathcal{F}_{\alpha}u$, may not be integrable in general, we cannot invert
the FRFT on $L^{1}$.
In \cite{CFGW}, we developed the harmonic analysis theoretical background that
addresses FRFT inversion issues of an $L^{1}$-signal via summability
techniques. Based on our earlier research and inspired by the above
literature, in this paper, we use this approximate inversion to study the
signal encryption from a new perspective.

\section{Mathematical background}

\label{sect:pre}

In this section, we recall properties of FRFT and analyze their numerical
adaptation. Define the chirp operator $\mathcal{M}_{\alpha}$ acting on
functions $\phi$ in $L^{1}(\mathbb{R})$ as follows:
\[
\mathcal{M}_{\alpha}\phi(x)=e^{\pi ix^{2}\cot\alpha}\phi(x).
\]
For $\alpha\neq n\pi$, let $A_{\alpha}$ be as in (\ref{defAa}). Then the FRFT
of $u\in L^{1}(\mathbb{R})$ can be written as
\begin{align}
(\mathcal{F}_{\alpha}u)(x)  &  =A_{\alpha}e^{i\pi x^{2}\cot\alpha}%
(\mathcal{F}e^{i\pi(\cdot)^{2}\cot\alpha}u)(x\csc\alpha)\nonumber\\
&  =A_{\alpha}\mathcal{M}_{\alpha}(\mathcal{FM}_{\alpha}u)(x\csc\alpha).
\label{eq:def}%
\end{align}
In view of (\ref{eq:def}), we see that the FRFT of a signal $u(t)$ can be
decomposed into four simpler operators, according to the diagram of Fig.
\ref{fig:com}:

\begin{enumerate}
[(i)]

\item chirp modulation, $v(t)=e^{\pi it^{2}\cot\alpha}u(t)$;

\item Fourier transform, $\hat{v}(x)=(\mathcal{F}v)(x)$;

\item scaling, $\tilde{v}(x)=\hat{v}(x\csc\alpha)$;

\item chirp modulation, $(\mathcal{F}_{\alpha}u)(x)=A_{\alpha}e^{\pi
ix^{2}\cot\alpha}\tilde{v}(x)$.
\end{enumerate}

\begin{figure}[t]
\centering
\begin{tikzpicture}[thick]
		\node (start){$u(t)$};
		\node[node distance=15mm, inner sep=0pt, right of=start] 
		(MA){$\bigotimes$};
\node[node distance=15mm, below of =MA] (p1){{$e^{i\pi t^2 cot\alpha}$}};
		\node[node distance=18mm, rectangle,draw,right of=MA] (FT){FT};
		\node[node distance=25mm, rectangle,draw,right of=FT] (SC){Scaling};
		\node[node distance=25mm, inner sep=0pt,right of=SC] 
		(MA2){$\bigotimes$};
\node[node distance=15mm, below of =MA2] (p2){$A_\alpha e^{i\pi x^2 cot\alpha}$};
		\node[node distance=21mm, right of=MA2] (end){$(\mathcal F _\alpha 
		u)(x)$};
		\draw[->](start)--(MA);
		\draw[->](MA)--node[above]{$v(t)$}(FT);
		\draw[->](FT)--node[above]{$\hat{v}(x)$}(SC);
		\draw[->](SC)--node[above]{$\tilde{v}(x)$}(MA2);
		\draw[->](MA2)--(end);
		\draw[->](p1)--(MA);
		\draw[->] (p2) -- (MA2);
		\end{tikzpicture}
\caption{The decomposition of the FRFT.}%
\label{fig:com}%
\end{figure}
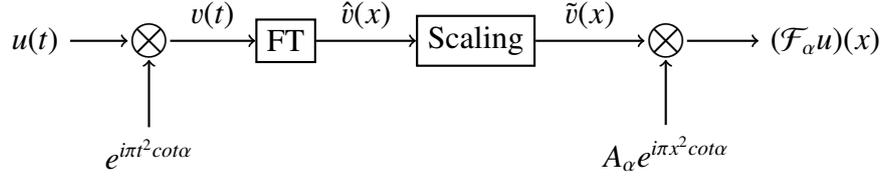

Generally speaking, the FRFT was studied on the Schwartz space $S(\mathbb{R})$
or on $L^{2}(\mathbb{R})$. Thanks to topological properties of these spaces,
the operator $\mathcal{F}_{\alpha}$ on $S(\mathbb{R})$ or $L^{2}(\mathbb{R})$
is unitary, invertible (with inverse transform $\mathcal{F}_{-\alpha}$) and
forms an additive group (i.e., $\mathcal{F}_{\alpha_{1}}\mathcal{F}%
_{\alpha_{2}}=\mathcal{F}_{\alpha_{1}+\alpha_{2}}$). However, many common
functions do not belong to $S(\mathbb{R})$ or $L^{2}(\mathbb{R})$. For
example, the following useful signal on the line
\begin{equation}
	u(t)=\left\{
	\begin{array}
		[c]{cc}%
		ne^{-i\pi x^{2}\cot\alpha}, & n\leq t<n+\frac{1}{n^{3}},\\
		0, & \text{otherwise.}%
	\end{array}
	\right.  \label{fun:1}%
\end{equation}
lies in $L^{1}(\mathbb{R})$ but not in $L^{2}(\mathbb{R})$. Using
(\ref{eq:fa}), we calculate the FRFT of this function:
\[
\left(  \mathcal{F}_{\alpha}u\right)  (x)=\frac{A_{\alpha}e^{i\pi x^{2}\cot\alpha}}{2\pi 
ix\csc\alpha}\sum\limits_{n=1}^{\infty}ne^{-2n\pi
	ix\csc\alpha}\left(
1-e^{-\frac{2\pi ix\csc\alpha}{n^{3}}}\right)  ,
\]
where $A_{\alpha}$ is as in (\ref{defAa}).

In the $L^{1}(\mathbb{R})$ setting, problems of convergence arise when
studying FRFT inversion. Given the FRFT of an $L^{1}$-signal of in fractional
Fourier domain, how to recover it back into time domain to be the original
signal? We naturally hope that%
\begin{equation}
(\mathcal{F}_{-\alpha}\mathcal{F}_{\alpha}u)(t)=\int_{-\infty}^{+\infty
}(\mathcal{F}_{\alpha}u)(x)K_{-\alpha}(x,t)\mathrm{d}x=u(t) \label{eq:in frft}%
\end{equation}
Unfortunately, when $u$ is integrable, one may not necessarily have that
$\mathcal{F}_{\alpha}u$ is integrable, so the integral in (\ref{eq:in frft})
might not make sense. In fact, $\mathcal{F}_{\alpha}u$ is nonintegrable in
general (cf., \cite[pp. 12]{Duo2001}). For example, let%
\begin{equation}
u(t)=e^{-\pi it^{2}}\left\vert t\right\vert ^{-1/2}\mathrm{rect}(t).
\label{eq:u}%
\end{equation}
Then $u\in L^{1}(\mathbb{R})$ but
\begin{equation}
U(x):=\left(  \mathcal{F}_{\pi/4}u\right)  (x)=2^{3/4}A_{\alpha}e^{i\pi x^{2}%
}\frac{C\left(  \sqrt{\frac{2^{3/2}}{\pi}\left\vert x\right\vert }\,\,\right)
}{\sqrt{\left\vert x\right\vert }} \label{eq:Fu}%
\end{equation}
and $U\notin L^{1}(\mathbb{R})$.
Then $\mathcal{F}_{-\alpha}\mathcal{F}_{\alpha}u$ does not make sense. Here
$\mathrm{rect}(t)$ is the rectangle function on the line
defined by
\[\mathrm{rect}(t)=\left\{
\begin{array}
[c]{cc}1, & \left\vert t\right\vert \leq1,\\
0, & \left\vert t\right\vert >1,
\end{array}
\right.  \]
and $\operatorname{C}(x)$ is the Fresnel cosine integral function
defined by \textrm{C}$(x)=\int_{0}^{x}\cos\frac{\pi s^{2}}{2}$\textrm{d}$s$.

In order to overcome the difficulty of non-integrability and recover the
original signal, we adopt the idea of inversion via summability means,
established in our earlier work \cite{CFGW}.

\begin{definition}
\label{def:M_Phi}Given $\Phi\in C_{0}(\mathbb{R})$ and $\Phi(0)=1$, a function
$u$, and $\varepsilon>0$ we define
\[
M_{\varepsilon,\Phi_{\alpha}}(u):=\int_{-\infty}^{+\infty}(\mathcal{F}%
_{\alpha}u)(x)K_{-\alpha}(x,\cdot)\Phi_{\alpha}(\varepsilon x)\mathrm{d}x,
\]
where
\[
\Phi_{\alpha}\left(  x\right)  :=\Phi\left(  x\csc\alpha\right)  .
\]
The expressions $M_{\varepsilon,\Phi_{\alpha}}(u)$ (with varying $\varepsilon
$) are called the $\Phi_{\alpha}$\emph{ means of the fractional Fourier
integral of }$u$.
\end{definition}

The following results concern FRFT approximate identities.

\begin{theorem}
\label{th:L1}If $\Phi,\mathcal{F}\Phi\in L^{1}(\mathbb{R})$ and $\left\Vert
\mathcal{F}\Phi\right\Vert _{L^{1}}=1$, then the $\Phi_{\alpha}$ means of the
fractional Fourier integral of $u$ are convergent to $u$ in the sense of
$L^{1}$ norm, that is,
\[
\lim_{\varepsilon\rightarrow0}\left\Vert M_{\varepsilon,\Phi_{\alpha}%
}(u)-u\right\Vert _{L^{1}}=0.
\]

\end{theorem}

\begin{theorem}
\label{th:pp} Suppose that $\Phi,\mathcal{F}\Phi$ lie in $L^{1}(\mathbb{R}) $
and that the function $\psi(x)=\underset{\left\vert t\right\vert
\geq\left\vert x\right\vert }{\sup}\left\vert (\mathcal{F}\Phi)\left(
t\right)  \right\vert $ is integrable over the line and $\left\Vert
\mathcal{F}\Phi\right\Vert _{L^{1}}=1$. Then the $\Phi_{\alpha}$ means of the
fractional Fourier integral of $f$ are a.e. convergent to $f$, that is,
\[
M_{\varepsilon,\Phi_{\alpha}}(u)(t)\rightarrow u\left(  t\right)
\]
as $\varepsilon\rightarrow0$ for almost all $t\in\mathbb{R}$.
\end{theorem}

Even if $\mathcal{F}_{\alpha}u$ is non-integrable, once multiplied by a smooth
cutoff, its inverse FRFT can be defined. Then, in view of theorems \ref{th:L1}
and \ref{th:pp}, we can reconstruct the original signal $u$ as a limit of the
$\Phi_{\alpha}$ means of $\mathcal{F}_{\alpha}u$ as $\varepsilon\rightarrow0$.

In engineering applications, it is necessary to calculate the discrete
fractional Fourier transform (DFRFT). It is not surprising that the numerical
implementation of the DFRFT is more complicated than that of the ordinary
discrete Fourier transform (DFT). At present, there are various types of fast
DFRFT algorithms with different processing methods and variable accuracy (cf.,
\cite{839980,757221,492554}). These form the basis for the successful
application of FRFT in signal processing. A basic point in the definition of
the discrete FRFT is its sufficient proximity to the continuous FRFT. To
recover the original signal from the fractional Fourier domain back to the
time domain, one usually tries to numerically calculate the FRFT of order
$-\alpha$ of (\ref{eq:Fu}). If we ignore the domain of the original signal in
the implementation of the algorithm, we will not be able to successfully
restore the signal. The essential reason is that (\ref{eq:Fu}) is not 
integrable.

\section{Signal encryption and decryption}

\label{sect:en}

\subsection{Encryption algorithm}

\label{subsect:en}

From the perspective of signal encryption, the difficulty of the FRFT
inversion problem of $L^{1}$ functions can be used to improve the security of
the encryption. Specifically, for a real-valued signal $u$, we first map it to
a signal $v $ which lies in $L^{1}(\mathbb{R})\backslash L^{2}(\mathbb{R})$,
and then through the fractional Fourier transform, we obtain the encrypted
signal $u^{e}=\mathcal{F}_{\alpha}v$; here and in the sequel, the superscript
$e$ indicates the encryption process and the superscript $d$ indicates the
decryption process. This kind of encryption based on FRFT usually needs to use
the inverse FRFT $\mathcal{F}_{-\alpha}$ when decrypting. However, as
mentioned above, the inverse transform $\mathcal{F}_{-\alpha}u^{e}$ does not
make sense if $u^{e}\notin L^{1}(\mathbb{R})$. This presents deciphering
complications, even if the secret key $\alpha$ is known.

The FRFT is a common and efficient tool in signal encryption. The difficulty
of deciphering can be enhanced by adding keys such as the multiple iterations
in \cite{iter}, combinations with the jigsaw transform \cite{jigsaw}, joint
transform correlators \cite{JntTranCor}, the region shift encoding
\cite{RegShEn}, chaotic maps \cite{chao}, multiple-phase codes \cite{RanPh},
etc. This paper only focuses on algorithms involving special properties of
FRFT on $L^{1}(\mathbb{R})$.

\begin{theorem}
\label{th:L1-2}For any bounded function $u$ and $\omega\in L^{1}%
(\mathbb{R})\setminus L^{2}(\mathbb{R})$ with $\omega(t)\neq0$ for
$t\in\mathbb{R}$, let $P_{\omega}u=u_{\omega}:=(u+M)\omega$, $M=1+\sup
_{t\in\mathbb{R}}\left\vert u(t)\right\vert $. Then $u_{\omega}\in
L^{1}(\mathbb{R})\setminus L^{2}(\mathbb{R})$.
\end{theorem}

\begin{proof}
Because $1\leq u(t)+M\leq2M$, then
\begin{align*}
\int_{-\infty}^{+\infty}\left\vert u_{\omega}(t)\right\vert \text{d}x  &
=\int_{-\infty}^{+\infty}\left\vert u(t)+M\right\vert |\omega(t)|\text{d}t\\
&  \leq2M\int_{-\infty}^{+\infty}\left\vert \omega(t)\right\vert
\text{d}t<+\infty,
\end{align*}
and
\begin{align*}
\int_{-\infty}^{+\infty}\left\vert u_{\omega}(t)\right\vert ^{2}\text{d}x  &
=\int_{-\infty}^{+\infty}\left\vert u(t)+M\right\vert ^{2}|\omega
(t)|^{2}\text{d}t\\
&  \geq\int_{-\infty}^{+\infty}\left\vert \omega(t)\right\vert ^{2}%
\text{d}t=+\infty.
\end{align*}
Then, we have $u_{\omega}\in L^{1}(\mathbb{R})$ and $u_{\omega}\not \in
L^{2}(\mathbb{R})$. The desired result is proved.
\end{proof}

Consider the following examples:
\begin{equation}
\begin{aligned}
\omega_{1}(t) = & \sum\limits_{i=1}^{n}\left\vert t-\tau_{i}\right\vert
^{-1/2}\chi_{\lbrack-k,k]}(t),\\
&\tau_{i}\in[-k,k],~,k\in\mathbb{R}^{+},~i=1,2,\ldots,n,
\end{aligned}
\label{eq:omega1}%
\end{equation}%
\begin{equation}
\omega_{2}(t)=\sum_{n=1}^{\infty}\left(  \sqrt{n}\chi_{(\frac{1}{n+1},\frac
{1}{n}]}( \vert t \vert )+\frac{\chi_{(n,n+1]}(\vert t \vert )}{(n+1)^{2}}\right)  , \label{eq:omega2}%
\end{equation}
where $\chi_{\lbrack a,b]}$ denotes the characteristic function of the
interval $[a,b]$.

The function in (\ref{eq:omega1}) is a linear combination of functions of type
(\ref{eq:u}), and obviously lies in $L^{1}(\mathbb{R})\setminus L^{2}%
(\mathbb{R})$. Also,
\[
\omega_{2}(t)=\left\{
\begin{array}
[c]{cc}%
\sqrt{n}, & \frac{1}{n+1}<\left\vert t\right\vert \leq\frac{1}{n},\\
\frac{1}{(n+1)^{2}} & n<\left\vert t\right\vert \leq n+1,
\end{array}
\right.  n=1,2,\ldots.
\]
We have%
\begin{align*}
\int_{\mathbb{R}}\left\vert \omega_{2}(t)\right\vert \mathrm{d}t  &
=2\sum\limits_{n=1}^{\infty}\int_{\frac{1}{n+1}}^{\frac{1}{n}}\sqrt
{n}\,\mathrm{d}t+2\sum\limits_{n=1}^{\infty}\int_{n}^{n+1}\frac{1}{(n+1)^{2}%
}\,\mathrm{d}t\\
&  =2\sum\limits_{n=1}^{\infty}\frac{1}{\sqrt{n}(n+1)}+2\sum\limits_{n=1}%
^{\infty}\frac{1}{(n+1)^{2}}<\infty
\end{align*}
and%
\begin{align*}
\int_{\mathbb{R}}\left\vert \omega_{2}(t)\right\vert ^{2}\mathrm{d}t  &
=2\sum\limits_{n=1}^{\infty}\int_{\frac{1}{n+1}}^{\frac{1}{n}}n\,\mathrm{d}%
t+2\sum\limits_{n=1}^{\infty}\int_{n}^{n+1}\frac{1}{(n+1)^{4}}\,\mathrm{d}t\\
&  \geq2\sum\limits_{n=1}^{\infty}\frac{1}{n+1}=\infty.
\end{align*}

\begin{remark}
Denote by $Q_{\omega}u_{\omega}:=\frac{u_{\omega}}{\omega}-M$, then it is easy
to see that $Q_{\omega}u_{\omega}=u$. For example, let $u(t)=\mathrm{rect}%
(t)$, $\omega=\omega_{1}$. Then $P_{\omega}u\in L^{1}(\mathbb{R})\setminus
L^{2}(\mathbb{R})$ and $Q_{\omega}u_{\omega}=u(t)=\mathrm{rect}(t)$.
\end{remark}

Let $u$ be a bounded signal function to be encrypted. In this paper, we always
assume that $u$ is a real-valued function. First, we randomly select an
$n$-dimensional sequence $\{\tau_{i}\}\in\lbrack-k,k]$ and then pick a
function $\omega$ as in (\ref{eq:omega1}) to map $u$ to $u_{\omega}$. By
Theorem \ref{th:L1-2}, $u_{\omega}\in L^{1}(\mathbb{R})\setminus
L^{2}(\mathbb{R})$. Then, through the FRFT of order $\alpha$, we get the
encrypted signal $u^{e}=\mathcal{F}_{\alpha}u_{\omega}$. In view of the chirp
decomposition of FRFT (Equ. (\ref{eq:def})), the encryption process can be
divided into the following steps, according to the diagram of Fig.
\ref{fig:en}:

\begin{figure}[t]
\centering
\includegraphics[width=0.8\linewidth]{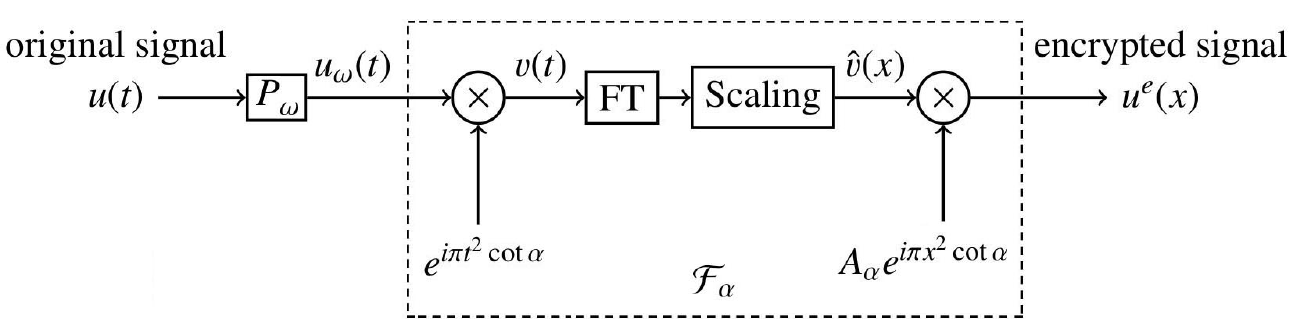}
\caption{The process of encryption algorithm.}%
\label{fig:en}%
\end{figure}

\begin{enumerate}
[(i)]

\item mapping $u$ into $L^{1}(\mathbb{R})\setminus L^{2}(\mathbb{R})$:
$u_{\omega}=P_{\omega}u$;

\item chirp modulation: $v(t)=e^{\pi it^{2}\cot\alpha}u_{\omega}(t)$;

\item Fourier transformation (with scale variation): \[\hat{v}(x)=\left(
\mathcal{F}v\right)  (x\csc\alpha);\]

\item chirp modulation: $u^{e}(x)=A_{\alpha}e^{\pi ix^{2}\cot\alpha}\hat
{v}(x)=\mathcal{F}_{\alpha}u_{\omega}$.
\end{enumerate}

In this way, we obtain the encrypted signal in the fractional Fourier domain,
which can be expressed as%
\begin{align}
u^{e}(x)&=(\mathcal{F}_{\alpha}P_{\omega}u)(x)\nonumber\\
&=\int_{-\infty}^{+\infty}K_{\alpha}(x,t) (u(t)+M)\omega(t) \,\mathrm{d}t,
\label{eq:en}%
\end{align}
where $M=1+\sup
_{t\in\mathbb{R}}\left\vert u(t)\right\vert $.

Here, the algorithm has two secret keys: the operator $P_{\omega}$ (including
sequence $\{\tau_{i}\}$) and the order $\alpha$ of FRFT. The keys here provide
high security. On the one hand, the assurance that $u^{e}(x)$ is not
integrable, provides instability in the reconstruction of the original signal
through the inverse FRFT, as mentioned above. On the other hand, it is known
that the operator $\mathcal{F}_{\alpha}$ is not continuous in the order
$\alpha$, that is, $\mathcal{F}_{\beta}$ may not map to $\mathcal{F}_{\alpha}$
as $\beta\rightarrow\alpha$. Therefore, it cannot be decrypted properly when
the order $\alpha$ is unknown.

\subsection{Decryption algorithm}

In spaces where an inverse transform is inoperative, in order to introduce the
decryption algorithm, we study inversion by adopting approximate identities.
We focus on two functions that give rise to special $\Phi_{\alpha}$ means
(Definition \ref{def:M_Phi}). Denote by
\begin{equation}
\left\{
\begin{aligned}
a_{\alpha}\left(  x\right)  &=e^{-2\pi\varepsilon\left\vert \csc\alpha
\right\vert \left\vert x\right\vert },\\
g_{\alpha}\left(x\right) & =e^{-4\pi^{2}\varepsilon x^{2}\csc^{2}\alpha}.
\end{aligned}
\right.
\label{ag}%
\end{equation}
The $\Phi_{\alpha}$ means
\[
M_{\varepsilon,a_{\alpha}}(u):=\int_{-\infty}^{+\infty}(\mathcal{F}_{\alpha
}u)(x)K_{-\alpha}(x,\cdot)a_{\alpha}(x)\mathrm{d}x
\]
and%
\[
M_{\varepsilon,g_{\alpha}}(u):=\int_{-\infty}^{+\infty}(\mathcal{F}_{\alpha
}u)(x)K_{-\alpha}(x,\cdot)g_\alpha (x)\mathrm{d}x
\]
are called the \emph{Abel mean }and\emph{ Gauss mean} of the fractional
Fourier integral of $u$, respectively. The following results are well-known.

\begin{lemma}
[\cite{StW,GL}]\label{pro:poisson}Let $\varepsilon>0$. Then

\begin{enumerate}
[(a)]

\item $\mathcal{F}\left[  e^{-2\pi\varepsilon\left\vert \,\cdot\,\right\vert
}\right]  \left(  x\right)  =\frac{1}{\pi}\frac{\varepsilon}{\varepsilon
^{2}+x^{2}}=:P_{\varepsilon}\left(  x\right)  $\quad(Poisson kernel);

\item $\mathcal{F}\left[  e^{-4\pi^{2}\varepsilon\left\vert \,\cdot
\,\right\vert ^{2}}\right]  \left(  x\right)  =\frac{1}{\left(  4\pi
\varepsilon\right)  ^{1/2}}e^{-x^{2}/4\varepsilon}=:G_{\varepsilon}\left(
x\right)  $\quad(Gauss-Weierstrass kernel).

\item $G_{\varepsilon},P_{\varepsilon}\in L^{1}(\mathbb{R})$;

\item $\int_{-\infty}^{+\infty}G_{\varepsilon}(x)\,\mathrm{d}x=\int_{-\infty
}^{+\infty}P_{\varepsilon}(x)\,\mathrm{d}x=1$.
\end{enumerate}
\end{lemma}

Combining with Theorem \ref{th:L1} and Theorem \ref{th:pp}, we can obtain the
following results immediately.

\begin{theorem}
\label{th:p-w}The Abel and Gauss means of the fractional Fourier integral of
$u^{e}$ (the encrypted signal) converge to $u_{\omega}$ in the sense of
$L^{1}$ norm and a.e., that is,

\begin{enumerate}[(i)]

\item $\lim\limits_{\varepsilon\rightarrow0}\left\Vert \mathcal{F}_{-\alpha
}(a_\alpha u^{e})-u_{\omega}\right\Vert _{L^{1}}=0,$

$\lim\limits_{\varepsilon\rightarrow0}\left\Vert \mathcal{F}_{-\alpha
}(g_\alpha u^{e})-u_{\omega}\right\Vert
_{L^{1}}=0;$

\item $\lim\limits_{\varepsilon\rightarrow0}\mathcal{F}_{-\alpha}%
(a_\alpha u^{e})(t)=u_{\omega}(t)$ for a. e. $t\in\mathbb{R},$

$\lim\limits_{\varepsilon\rightarrow0}\mathcal{F}_{-\alpha}(g_\alpha 
u^{e})(t)=u_{\omega}(t)$ for a. e.
$t\in\mathbb{R}.$
\end{enumerate}
\end{theorem}

\begin{remark}
From the expressions of the kernel functions in (\ref{ag}), it is obvious that
the numerical implementation of function $a_{\alpha}\left(  x\right)  $ is
less demanding than that of function $g_{\alpha}\left(  x\right)  $. So in the
following encryption algorithm, we choose the function $a_{\alpha}\left(
x\right)  $ and use the fractional Abel means to approximate the original signal.
\end{remark}

Given an encrypted signal $u^{e}$, we can obtain the associated decrypted
signal $u^{d}$ by taking the Abel means of the fractional Fourier integral of
$u^{e}$ with $\varepsilon$ small enough. The decryption process may be divided
into the following steps according to the diagram of Fig. \ref{fig:de}:

\begin{figure}[t]
\centering
\includegraphics[width=0.8\linewidth]{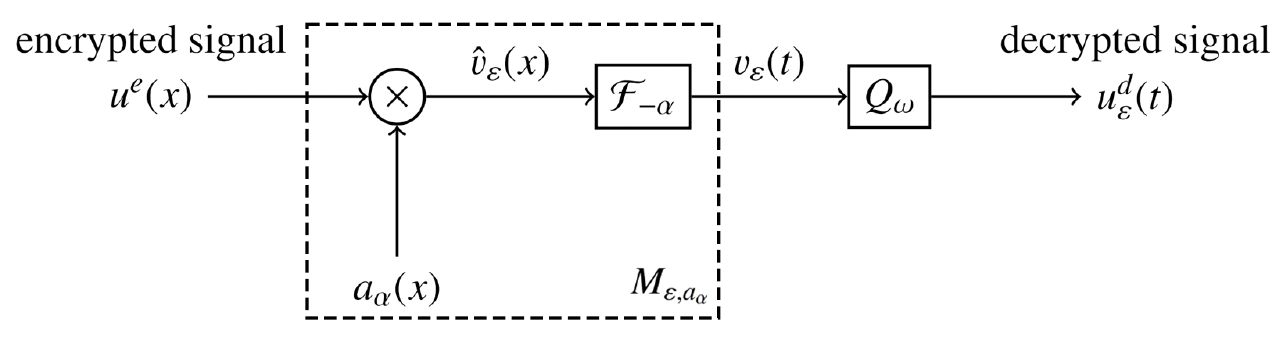}
\caption{The process of decryption algorithm.}%
\label{fig:de}%
\end{figure}

\begin{enumerate}
[(i)]

\item multiplication by the Abel function $ a_{\alpha}\left(  x\right)  $:
$\hat{v}_{\varepsilon}(x)=u^{e}(x)a_\alpha (x);$

\item inverse FRFT: $v_{\varepsilon}(t)=(\mathcal{F}_{-\alpha}\hat
{v}_{\varepsilon})(t);$

\item action by $Q_{\omega}$: $u_{\varepsilon}^{d}(t)=Q_{\omega
}v_{\varepsilon}(t)={v_{\varepsilon}(t)}/{\omega(t)}-M$.
\end{enumerate}

This process can be simply described as
\begin{align}
u_{\varepsilon}^{d}(t) &=Q_{\omega}\left[  M_{\varepsilon,a_{\alpha}}u\right](t)\nonumber\\
&=Q_{\omega}[\mathcal{F}_{-\alpha}(a_\alpha u^{e})](t).\label{eq:de}%
\end{align}

Through the aforementioned method, we can restore the encrypted signal
$u^{e}(x)$ from the fractional Fourier domain back to the time domain. Since
$u$ is a real-valued signal, the amplitude of the signal $u_{\varepsilon}^{d}$
can approximate the original signal $u$ as $\varepsilon\rightarrow0$, in view
of Theorem \ref{th:p-w} (as in Fig. \ref{fig:Ozak2}). In other words, for
$\varepsilon$ sufficiently small, the error between the decrypted signal
$u_{\varepsilon}^{d}$ and the original signal $u$ can be arbitrarily small.
Furthermore, accuracy is improved when the parameter $\varepsilon$ gets
smaller.
Replacing $u^{e}(t)$ by the ``Abel average" $u^{e}(t) a_{\alpha}\left(  t\right)  $ yields improved smoothness and results in fewer
discretization errors.

\section{Simulation examples and applications in audio encrytion}

\begin{figure}[t]
	\centering
	\includegraphics[width=\linewidth]{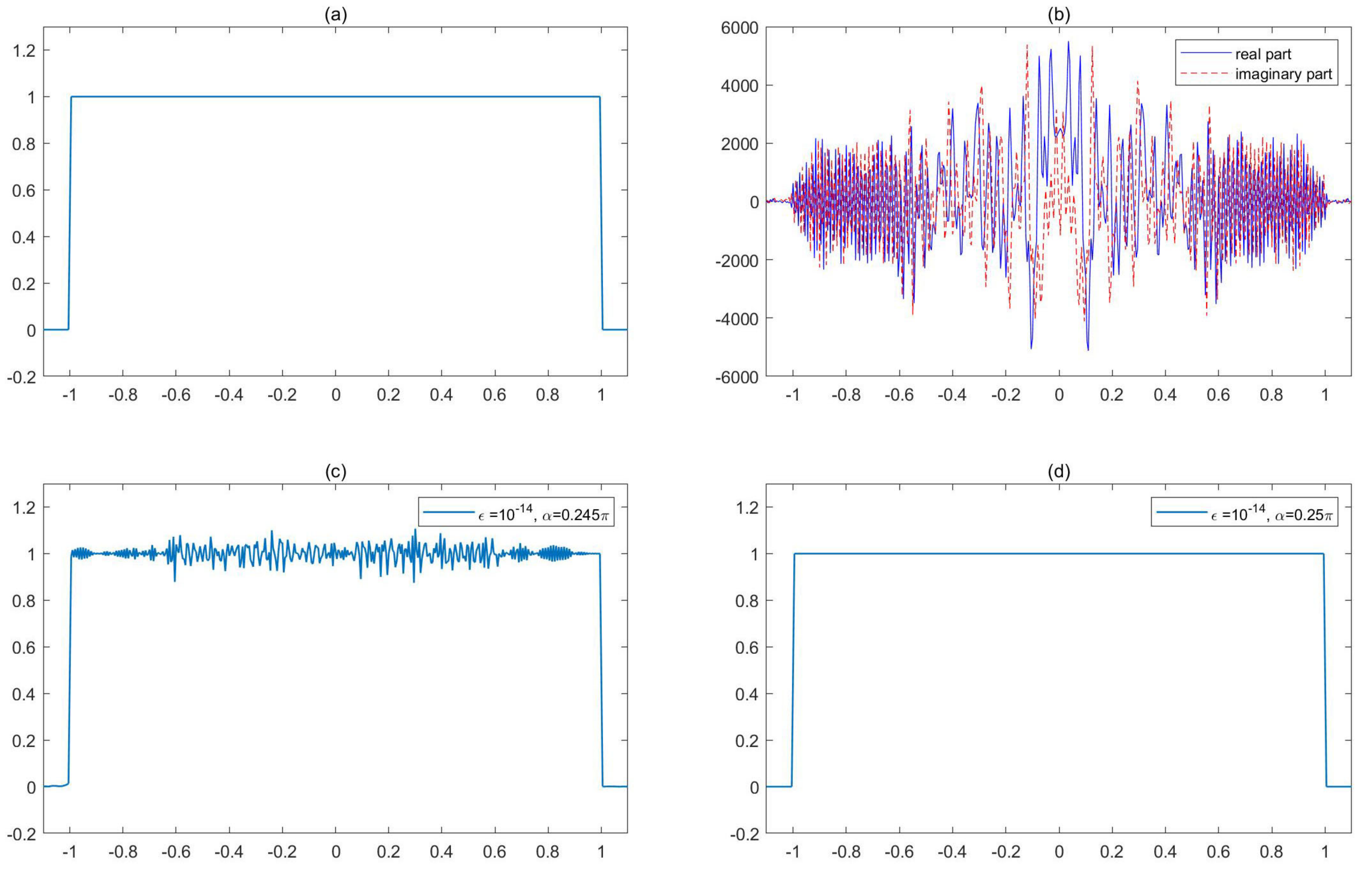}
	\caption{(a) The original signal $\mathrm{sect\;}(t)$; (b) the real and 
	imaginary part graphs of the encrypted
		signal $u^{e}(t)$; the decrypted signal $u^{d}$ with (c) incorrect key
		$\alpha= 0.245\pi$ and (d) correct key $\alpha= 0.25\pi$ for 
		$\varepsilon
		=10^{-14}$. }%
	\label{fig:rect}%
\end{figure}

\begin{figure}[t]
	\centering
	\includegraphics[width=\linewidth]{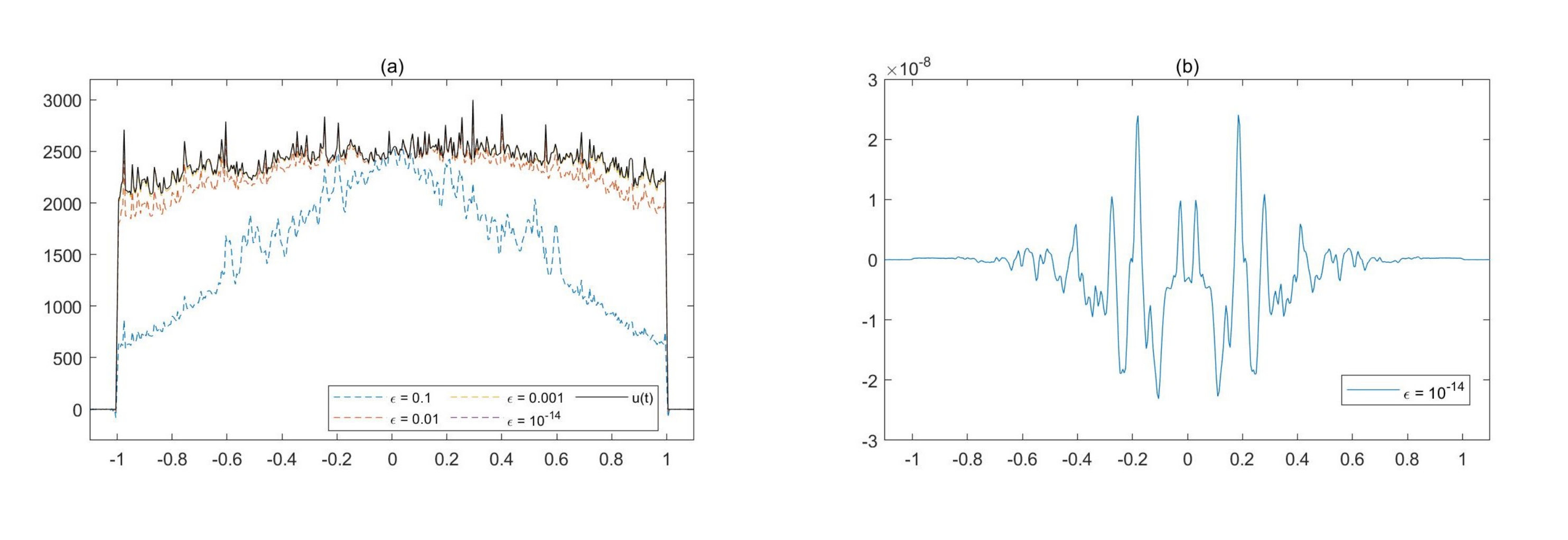}
	\caption{(a) The decrypted signal $u_{\varepsilon,\omega}(t)$ obtained by 
	Abel means with various	$\varepsilon$; (b) the error for 
	$\varepsilon=10^{-14}$.}%
	\label{fig:Ozak2}%
\end{figure}

We take a unit rectangular signal $u(t)=\mathrm{rect}(t)$ as an example. Here
we randomly select a sequence $\{\tau_{i}\}\in\lbrack-1.1,1.1]$ and let
$\omega(t)=\sum\limits_{i=1}^{n}\left\vert t-\tau_{i}\right\vert ^{-1/2}%
\chi_{\lbrack-1.1,1.1]}$ as in (\ref{eq:omega1}). Denote by
\[
u_{\omega}(t):=P_{\omega}u(t)=\sum\limits_{i=1}^{n}\left\vert t-\tau
_{i}\right\vert ^{-1/2}\mathrm{rect}(t).
\]
Then $u_{\omega}\in L^{1}(\mathbb{R})$ but $u_{\omega}\notin L^{2}%
(\mathbb{R})$. Take the fractional order $\alpha=\pi/4$ as a secret key.
Through the $\pi/4$-th FRFT of $u_{\omega}$, we get the encrypted signal
$u^{e}=\mathcal{F}_{\pi/4}u_{\omega}$ (see Fig. \ref{fig:rect} (b)) in the
fractional Fourier domain. Similar to (\ref{eq:u}), we can see that
$u^{e}\notin L^{1}(\mathbb{R})$ and the inverse FRFT%
\begin{equation}
\int_{-\infty}^{+\infty}u^{e}(x)K_{-\pi/4}(x,t)\mathrm{d}x \label{ex}%
\end{equation}
do not make sense.

In order to recover the original signal $u(t)$, we should use the
approximating method, that is, take the Abel means of the integral (\ref{ex})
\begin{equation}
u_{\varepsilon,\omega}(t)=A_{-\alpha}\int_{-\infty}^{+\infty}u^{e}%
(x)K_{-\pi/4}(x,t)a_\alpha (x) \mathrm{d}x. \label{eq:abel}%
\end{equation}
By Theorem \ref{th:p-w}, we know that $u_{\varepsilon,\omega}(t)\rightarrow
u_{\omega}(t)$ for a.e. $t\in\mathbb{R}$ as $\varepsilon\rightarrow0$, as
shown in Fig. \ref{fig:Ozak2} (a). Fig. \ref{fig:rect} (d) shows the decrypted
signal $u^{d}(t)$ obtaining by take Abel means of the integral (\ref{ex}) with
$\varepsilon=10^{-14}$. As shown in Fig. \ref{fig:Ozak2} (b), the numerical 
accuracy achieves the order of 
magnitude of $10^{-8}$.  In summary, the encryption
method proposed in this paper guarantees the security of the encryption
process regardless the way of the decryption or the security of secret keys.

Next, we applied the encryption strategy introduced above to encrypt 
the audio signal. Let's take Beethoven's famous piano music ``For Elise" (an 8 
seconds clip) as an 
example. We randomly select a sequence  $\{\tau_{i}\}\in [0,8]$ and 
$\alpha=\pi/4$ as the secret keys. Let
\[
\omega(t)=\sum\limits_{i=1}^{n}\left\vert t-\tau_{i}\right\vert ^{-1/2} 
\chi_{[0,8]}(t)
\]
 as in (\ref{eq:omega1}).  In the light of the algorithms shown in 
Fig. \ref{fig:en} and Fig. \ref{fig:de}, we get the encrypted and decrypted 
audio by using the FRFT algorithm based on FFT (Fast Fourier Transform). The 
waveform of the encrypted and decrypted audio signal ``For 
Elise" can be found in Fig. \ref{fig:AE}.

\begin{figure}[!t]
	\centering
	  \subfigure[The original audio signal $u$]{
	\includegraphics[width=\linewidth]{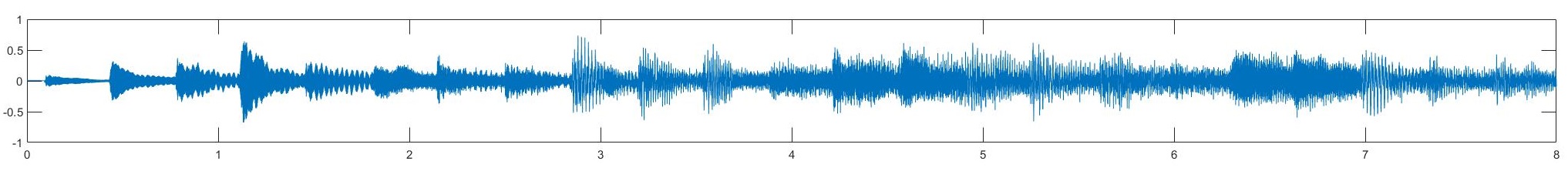}}
	  \subfigure[The single encrypted audio signal $P_\omega u$]{
	\includegraphics[width=\linewidth]{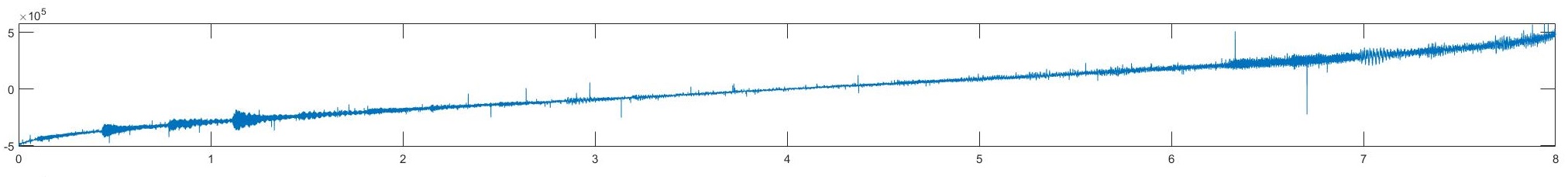}}
	  \subfigure[The double encrypted audio signal $\mathcal{F}_{\pi/4} 
	  P_\omega 
	  u$]{
	\includegraphics[width=\linewidth]{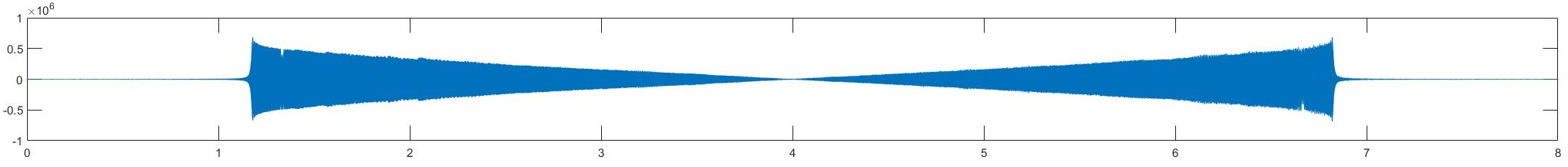}}
	  \subfigure[The decrypted audio signal with correct key $\alpha=0.25\pi$]{
	\includegraphics[width=\linewidth]{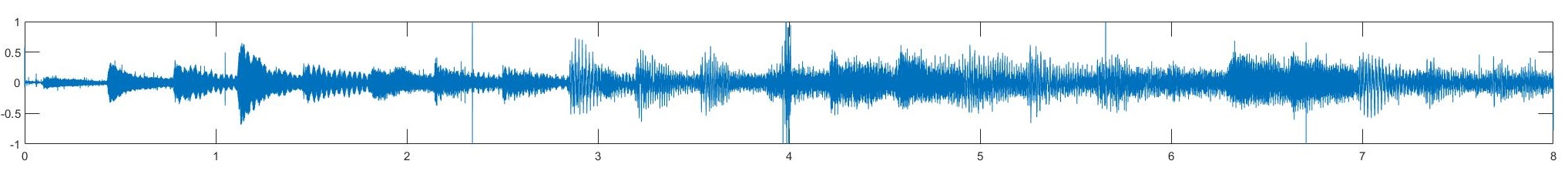}}
	  \subfigure[The decrypted audio signal with wrong key $\alpha=0.245\pi$]{
	\includegraphics[width=\linewidth]{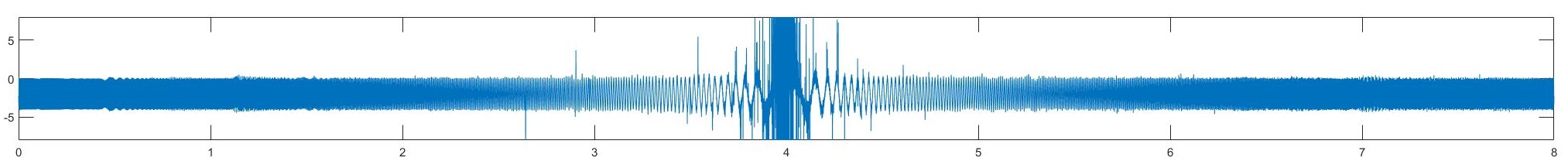}}
	\caption{Waveform of the encrypted and decrypted audio signal ``For 
	Elise".}
	\label{fig:AE}%
\end{figure}

\section{Connections with Fractional Fourier multipliers}

In signal encryption, FRFT is often used in combination with other transforms
or operators, as mentioned at the beginning of Section \ref{subsect:en}.
Fourier multipliers are operators defined by altering the Fourier transform by
multiplication; these play an important role in mathematical analysis and
signal processing. The authors' previous work \cite{CFGW} introduced Fourier
multipliers in the FRFT context. In this section, we combine these with the
encryption algorithm described in Section \ref{sect:en} to a new multiple encryption.

Let $1\leq p\leq\infty$ and $m_{\alpha}\in L^{\infty}(\mathbb{R})$. Define the
operator $T_{m_{\alpha}}$ as%
\begin{equation}
\mathcal{F}_{\alpha}\left(  T_{m_{\alpha}}f\right)  \left(  x\right)
=m_{\alpha}\left(  x\right)  \left(  \mathcal{F}_{\alpha}f\right)  \left(
x\right),\;\forall f\in L^{2}(\mathbb{R})\cap L^{p}(\mathbb{R}).
\label{def:ma}%
\end{equation}
The function $m_{\alpha}$ is called \emph{the }$L^{p}$\emph{ Fourier
multiplier of order }$\alpha$, if there exist a constant $C_{p,\alpha}>0$ such
that
\begin{equation}
\left\Vert T_{m_{\alpha}}f\right\Vert _{p}\leq C_{p,\alpha}\left\Vert
f\right\Vert _{p},\quad\forall f\in L^{2}(\mathbb{R})\cap L^{p}(\mathbb{R}).
\label{eq:ma}%
\end{equation}

\noindent As $L^{2}(\mathbb{R})\cap L^{p}(\mathbb{R})$ is dense in
$L^{p}(\mathbb{R})$, there is a unique bounded extension of $T_{m_{\alpha}}$
in $L^{p}(\mathbb{R})$ satisfying (\ref{eq:ma}). This extension is also
denoted by $T_{m_{\alpha}}$ and
\[
T_{m_{\alpha}}f=\mathcal{F}_{-\alpha}\left[  m_{\alpha}\left(  \mathcal{F}%
_{\alpha}f\right)  \right]  .
\]

Given an $L^{2}$-signal $u$ to be encrypted. Let $T_{m_{\beta}}$
be the operator associated with a fractional $L^{p}$ multipler $m_{\beta}$.
Then $T_{m_{\beta}}u\in L^{2}(\mathbb{R})$. Next, repeat the encryption
process as in Fig. \ref{fig:en} for $T_{m_{\beta}}u$ and we get the encrypted
signal $u^{e}$. According to the diagram of Fig. \ref{fig:multi}, $u^{e}$ can
be expressed as
\[
u^{e}=\mathcal{F}_{\alpha}\left[  P_{\omega}\left(  T_{m_{\beta}}u\right)
\right]  .
\]

\begin{figure}[t]
\centering
\includegraphics[width=0.5\linewidth]{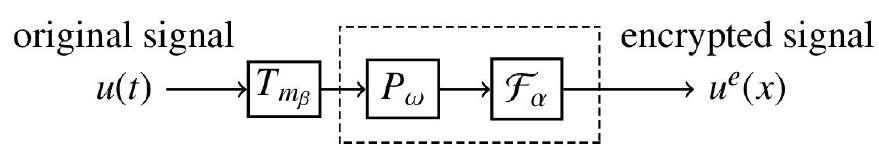}
\caption{The process of encryption algorithm}%
\label{fig:multi}%
\end{figure}

In view of Fig. \ref{fig:de}, the decryption process is shown in Fig.
\ref{fig:dmulti} and
\begin{align*}
u_{\varepsilon}^{d} &=T_{m_{\beta}}^{-1}\left(  Q_{\omega}M_{\varepsilon
,a_{\alpha}}u^{e}\right)  \\
&=\mathcal{F}_{-\beta}m_{\beta}^{-1}\mathcal{F}%
_{\beta}\left[  Q_{\omega}\mathcal{F}_{-\alpha}\left( a_\alpha (x)%
u^{e}\right)  \right]  .
\end{align*}

\begin{figure}[t]
\centering
\includegraphics[width=0.7\linewidth]{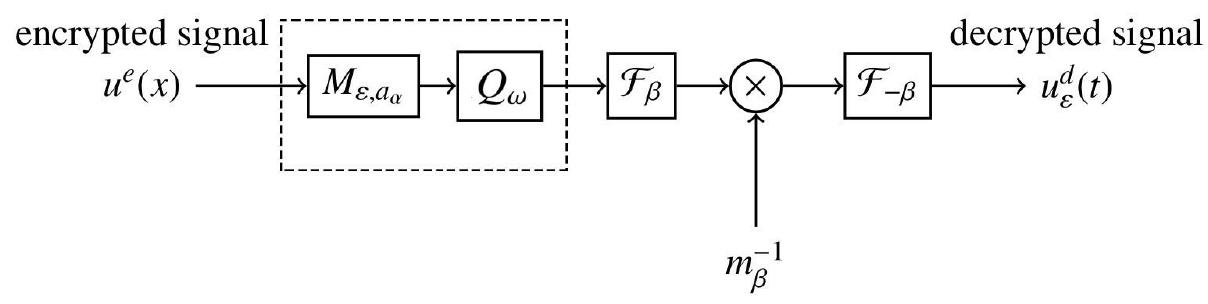}
\caption{The process of decryption algorithm}%
\label{fig:dmulti}%
\end{figure}

In view of (\ref{def:ma}), many important fractional integral operators can be
expressed in terms of fractional $L^{p}$ multiplier, for example the
fractional Hilbert transform. Recall that the classical Hilbert transform is
defined as
\begin{equation}
(\mathcal{H}u)(t) =\mathrm{p.v.~}\frac{1}{\pi}\int_{-\infty}^{+\infty}%
\frac{u(\tau)}{t-\tau}\mathrm{d}\tau. \label{eq:CHT}%
\end{equation}
The Hilbert transform of order $\beta$ is defined as (cf., \cite{zayed})%
\begin{equation}
(\mathcal{H}_{\beta}u)( t ) =\mathrm{p.v.~}\frac{e^{-i\pi t^{2}\cot\beta}}%
{\pi}\int_{-\infty}^{+\infty}\frac{u(\tau)e^{i\pi\tau^{2}\cot\beta}}{t-\tau
}\mathrm{d}\tau. \label{eq:HT}%
\end{equation}
For $1<p<\infty$, the operator $\mathcal{H}_{\beta}$ is bounded from
$L^{p}(\mathbb{R})$ to $L^{p}(\mathbb{R})$. By \cite[Theorem 4]{zayed}, we see
that
\[m_{\beta}=-i\mathrm{sgn}\left(  (\pi-\beta)\omega^{\prime}\right)
\]
 is a fractional $L^{p}$ multiplier and the associated operator $T_{m_{\beta}}$ is
the fractional Hilbert transform, that is,%
\begin{equation}
(\mathcal{F}_{\beta}\mathcal{H}_{\beta}u)\left(  \omega^{\prime}\right)
=-i\mathrm{sgn}\left(  (\pi-\beta)\omega^{\prime}\right)  \left(
\mathcal{F}_{\beta}u\right)  \left(  \omega^{\prime}\right)  . \label{eq:mHT}%
\end{equation}
Without loss of generality, assume that $\beta\in(0,\pi)$. It can be seen from
(\ref{eq:mHT}) that the Hilbert transform of order $\beta$ is a phase-shift
converter that multiplies the positive portion in $\beta$-th fractional
Fourier domain of signal $u$ by $-i$, that is, maintaining the same amplitude,
shifts the phase by $-\pi/2$, while the negative portion of $\mathcal{F}%
_{\beta}u$ is shifted by $\pi/2$. As shown in Fig. \ref{fig:HT}.

\begin{figure}[!thb]
  \centering
  \subfigure[]{
  \begin{tikzpicture}[>=stealth,thick,scale=0.7]			
	\draw [->] (0,-3,0) -- (0,3,0) node [at end, above] {Im $U$};
    \draw [->] (0,0,-3) -- (0,0,3) node [at end, below] {Re $U$};
	\draw [fill=blue!20,draw=blue!70!red]  (1,0,0) -- (1,0,1.5) -- 
	(2.5,0,1.5)--(2.5,0,0);
	\draw [draw=blue!50!green,fill=blue!50!green!20] (-1,0,0) -- (-1,0,1.5) -- 
	(-2.5,0,1.5)--(-2.5,0,0);
	\draw [->] (-3,0,0) -- (3,0,0) node [at end, right] {$\omega'$};
  \end{tikzpicture}
  }
  \subfigure[]{
  \begin{tikzpicture}[>=stealth,thick,scale=0.7]
	\draw [->] (0,-3,0) -- (0,3,0) node [at end, above] {Im $V$};
	\draw [->] (0,0,-3) -- (0,0,3) node [at end, below] {Re $V$};
	\draw [fill=blue!20,draw=blue!70!red,rotate around x=90]  (1,0,0) -- 
	(1,0,1) -- (2.5,0,1)--(2.5,0,0);
	\draw [draw=blue!50!green,fill=blue!50!green!20,rotate around x=-90] 
	(-1,0,0) -- (-1,0,1) -- (-2.5,0,1)--(-2.5,0,0);
	\draw [->] (-3,0,0) -- (3,0,0) node [at end, right]{$\omega'$};
  \end{tikzpicture}
  }
  \subfigure[]{
  \begin{tikzpicture}[>=stealth,thick,scale=0.7]
	\draw[->](-3,0)--(3,0) node[right]{$t$};
	\draw[->](0,-3)--(0,3)node[above]{$\omega$};
	\draw[->,rotate=60,blue](-3,0)--(3,0) node[right]{$\omega'$};
	\draw[->,thin,red] (1.5,0) arc (0:60:1.5);
	\draw[->,thin] (0.6,0) arc (0:90:0.6);
	\node[right] at 	
	(1.2,0.9){\footnotesize$\textcolor{red}{\mathcal{F}_\beta}$};
	\node[right] at (0.4,0.4){\footnotesize$\mathcal{F}_{\frac{\pi}{2}}$};
  \end{tikzpicture}
  }
\caption{Phase-shifting effect of the $\beta$th-order Hilbert transform (
	(a) the original signal: $U=(\mathcal{F}_{\beta}u)(\omega^{\prime})$; 
	(b) after Hilbert transform of order $\beta$: 
	$V=(\mathcal{F}_{\beta}\mathcal{H}_{\beta}(u)(\omega^{\prime})$;
	(c) rotation of the time-frequency plane).}%
\label{fig:HT}%
\end{figure}
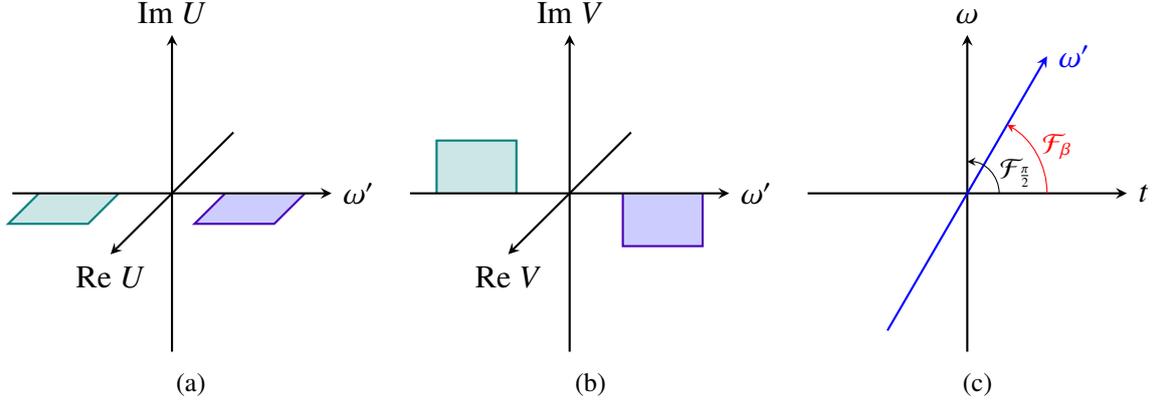

Take $T_{m_{\beta}}=\mathcal{H}_{\beta}$ and $\omega$ as in (\ref{eq:omega1})
as an example. The encryption process shown in Fig. \ref{fig:multi} can be
divided into the following steps:

\begin{enumerate}
[(i)]

\item phase shifting, $u_{1}=\mathcal{H}_{\beta}u$: shifting the phase of
positive portion of the original signal $u$ in $\beta$-th Fourier domain by
$-\pi/2$, while shifting the phase of negative portion of $\mathcal{F}_{\beta
}u$ by $\pi/2$;

\item adjusting amplitude, $u_{2}=P_{\omega}u_{1}$: enlarging the amplitude
nearby $\tau_{i}$, $i=1,2,\ldots,n$;

\item rotation of time-frequency plane, $u^{e}=\mathcal{F}_{\alpha}u_{2}$:
counterclockwise rotation with an angle of $\alpha$ of signal $u_{2}$ from the
time axis to $x$-axis around the origin in the time-frequency plane.
\end{enumerate}

In this way, we triply encrypt the signal. If any one of the keys is
erroneous, the original signal cannot be reconstructed. In addition, there are
various other options for $m_{\alpha}$ and $T_{m_{\alpha}}$, such as
fractional Poisson integral operator, fractional Gauss-Weierstrass integral
operator and so on.

\section{Conclusion}

FRFT is a powerful tool widely used in signal processing. In this work, we
elaborated the impact provided by altering the domain of the signals. We
pointed out that the FRFT of an $L^{1}$-signal is usually noninvertible. Then
we provided a double encryption algorithm based on the different properties of
FRFT in $L^{1}(\mathbb{R})$ and $L^{2}(\mathbb{R})$ spaces. We mapped the
to-be-encrypted signal to a ``bad" $L^{1}$-signal and obtained an encrypted
signal by applying the FRFT. Here, the two keys are the function $\omega$
and the order $\alpha$. With the help of Abel and Gauss means of FRFT, we used
the idea of identity approximation to recover the encrypted signal from the
fractional Fourier domain back to the time domain. 

On the one hand, according 
to the encryption algorithm in this paper, the original signal $u$ is 
multiplied by the 
secret function $\omega$, and then the FRFT is performed.
It is known from the convolution theorem (see \cite{zay982}) that
\[
\mathcal{F}_\alpha(u \omega)(x)=A_{-\alpha}e^{i \pi  x^2 \cot\alpha}\left( 
U_\alpha \ast \Omega_{\alpha} \right)(x),
\]
where $U_\alpha (t)= e^{-i \pi  t^2 \cot\alpha} \mathcal F_\alpha (u)(t)$, 
$\Omega_\alpha (t)=  e^{-i \pi  t^2 \cot\alpha} \mathcal F_\alpha (\omega 
e^{-i \pi (\cdot)^2 \cot\alpha})(t)$.
Namely, the product of a signal in the time domain becomes a convolution in  
fractional Fourier domain. Obviously, the convolution operation greatly 
increases the 
difficulty for the decipher to separate $u$ from $\omega$. 
Further,  by the Heisenberg's uncertainty principle (see \cite{Shin} and 
\cite[pp. 158]{StR}), it is impossible for a signal to have compact support in 
both time domain and fractional Fourier domain.
A function and its FRFT cannot both be essentially localized. Somewhat more 
precisely, if the ``preponderance" of the mass of a function is concentrated in 
an interval of finite length, then the preponderance of the mass of its FRFT 
cannot lie in an interval of finite length.
This means that the decipher cannot obtain all the characteristics of the 
original signal $u$ and the secret function  $\omega$ from the encrypted signal 
$u^e$,  which increases the difficulty of deciphering. 
On the other hand, using the FRFT $L^{1}$ theory established in \cite{CFGW}, 
the encryption algorithm in this paper not only
further improves security, but also ensures the feasibility and accuracy of
decryption. This can be shown in the simulation examples and applications in 
audio encrytion. Finally we studied the general idea of signal encryption 
combined
with fractional Fourier multipliers, and looked at the fractional Hilbert
transform as an example.

For the sake of simplicity, this paper only studied the related problems of
one-dimensional signals. In fact, using a similar idea, one may establish the
$L^{1}$ identity approximation theory of high-dimensional FRFT and apply it to
image encryption. The idea in  Theorem \ref{th:L1-2},
 combined with \cite[Theorem 1.12]{Duo2001},
 seems to provide a new design of an encryption algorithm even in the case of
  the classical Fourier transform.
  
\section*{Acknowledgments}
This work was partially supported by the National Natural Science Foundation of 
China (Nos. 12071197, 11701251 and 11771195), the Natural Science Foundation of 
Shandong Province (Nos. ZR2017BA015 and ZR2019YQ04),  a Simons Foundation 
Fellows Award (No. 819503) and a Simons Foundation Grant (No. 624733).

\bibliography{ref}

\begin{thebibliography}{99}
	
	
	\bibitem {Bernardo94}L.~M. Bernardo, O.~D.~D. Soares,
	\href{http://josaa.osa.org/abstract.cfm?URI=josaa-11-10-2622}{Fractional
		Fourier transforms and imaging}, J. Opt. Soc. Am. A \newblock  11~(10) 
		(1994) 2622--2626.
	
	\bibitem {839980}C.~{Candan}, M.~A. {Kutay}, H.~M. {Ozaktas}, The discrete
	fractional {F}ourier transform, IEEE Trans. Signal Process. \newblock  
	48~(5)
	(2000) 1329--1337.
	
	\bibitem{CFGW} W.~{Chen}, Z.~{Fu}, L.~{Grafakos}, Y.~{Wu},
	\href{https://doi.org/10.1016/j.acha.2021.04.004}{Fractional 
	{F}ourier	transforms on ${L}^{p}$ and applications}, Appl. Comput. 
	Harmon. Anal. \newblock 55~(2021) 71--96.
	
	\bibitem {RegShEn}L.~Chen, D.~Zhao, F.~Ge,
	\href{http://www.sciencedirect.com/science/article/pii/S003040181000012X}{Gray
		images embedded in a color image and encrypted with {F}{R}{F}{T} and 
		{R}egion
		{S}hift {E}ncoding methods}, Opt. Commum. \newblock  283~(10) (2010) 
		2043 -- 2049.
	
	\bibitem {D2001}I.~Djurovic, S.~Stankovic, I.~Pitas,
	\href{http://www.sciencedirect.com/science/article/pii/S1084804500901280}{Digital
		watermarking in the fractional Fourier transformation domain}, J. Netw.
	Comput. Appl. \newblock  24~(2) (2001) 167 -- 173.
	
	\bibitem {Duo2001}J.~Duoandikoetxea, Fourier analysis, Vol.~29 of Graduate
	Studies in Mathematics, American Mathematical Society, Providence, RI, 2001.
	
	\bibitem {jigsaw}B.~Hennelly, J.~T. Sheridan,
	\href{http://ol.osa.org/abstract.cfm?URI=ol-28-4-269}{{O}ptical image
		encryption by random shifting in fractional {F}ourier domains}, Opt. 
		Lett.
	\newblock  28~(4) (2003) 269--271.
	
	\bibitem{GL}
	L.~Grafakos, 
	\href{https://doi.org/10.1007/978-1-4939-1194-3}{Classical
			{F}ourier analysis}, {Graduate Texts in Mathematics}, vol. 
			249,
	Springer, New York, 3rd ed., 2014.
	
	\bibitem {Kerr1988}F.~H. Kerr,
	\href{http://www.sciencedirect.com/science/article/pii/0022247X88900947}{Namias'
		fractional Fourier transforms on ${L}^{2}$ and applications to 
		differential
		equations}, J. Math. Anal. Appl. \newblock  136~(2) (1988) 404 -- 418.
	
	\bibitem {KF1988}F.~H.~Kerr, A distributional approach to Namias' fractional
	Fourier transforms, Proc. Roy. Soc. Edinburgh Sect. A \newblock  108 (1988) 
	133--143.
	
	\bibitem {iter}Z.~Liu, S.~Liu,
	\href{http://www.sciencedirect.com/science/article/pii/S0030401807003240}{Double
		image encryption based on iterative fractional {F}ourier transform}, 
		Opt.
	Commun. \newblock  275~(2) (2007) 324 -- 329.
	
	\bibitem {Liu97}S.~Liu, H.~Ren, J.~Zhang, X.~Zhang,
	\href{http://ao.osa.org/abstract.cfm?URI=ao-36-23-5671}{Image-scaling 
	problem
		in the optical fractional Fourier transform}, Appl. Opt. \newblock  
		36~(23)
	(1997) 5671--5674.
	
	\bibitem {RanPh}S.~Liu, L.~Yu, B.~Zhu,
	\href{http://www.sciencedirect.com/science/article/pii/S0030401800010932}{Optical
		image encryption by cascaded fractional {F}ourier transforms with 
		random phase
		filtering}, Opt. Commun. \newblock  187~(1) (2001) 57 -- 63.
	
	\bibitem {Lohmann93}A.~W. Lohmann,
	\href{http://josaa.osa.org/abstract.cfm?URI=josaa-10-10-2181}{Image 
	rotation,
		wigner rotation, and the fractional Fourier transform}, J. Opt. Soc. 
		Am. A
	\newblock  10~(10) (1993) 2181--2186.
	
	\bibitem {MK1987}A.~C. {M}c{B}ride, {K}err {F}.~{H}., On namias's fractional
	Fourier transforms, IMA J. Appl. Math. \newblock  39 (1987) 159--175.
	
	\bibitem {M2002}T.~{Musha}, H.~{Uchida}, M.~{Nagashima}, Self-monitoring 
	sonar
	transducer array with internal accelerometers, IEEE J. Oceanic Eng.
	\newblock  27~(1) (2002) 28--34.
	
	\bibitem {namias}V.~Namias, The fractional order Fourier transform and its
	application to quantum mechanics, IMA J. Appl. Math. \newblock  25 (1980) 
	241--265.
	
	\bibitem {N2003}V.~A. Narayanan, K.~Prabhu,
	\href{http://www.sciencedirect.com/science/article/pii/S0141933103001133}{The
		fractional Fourier transform: theory, implementation and error 
		analysis},
	Microprocess. Microsy. \newblock  27~(10) (2003) 511 -- 521.
	
	\bibitem {Ozak}H.~M. Ozaktas, Z.~Zalevsky, M.~Kutay-Alper, The fractional
	Fourier transform : with applications in optics and signal processing, 
	Wiley,  New	York, 2001.
	
	\bibitem {757221}{S.-C. Pei}, {M.-H. Yeh}, {C.-C. Tseng}, Discrete 
	fractional
	{F}ourier transform based on orthogonal projections, IEEE Trans. Signal
	Process. \newblock 47~(5) (1999) 1335--1348.
	
	\bibitem {PhysRevLett72}M.~G. Raymer, M.~Beck, D.~McAlister,
	\href{https://link.aps.org/doi/10.1103/PhysRevLett.72.1137}{Complex 
	wave-field
		reconstruction using phase-space tomography}, Phys. Rev. Lett. 
		\newblock  72
	(1994) 1137--1140.
	
	\bibitem {Y2003}I.~{Samil Yetik}, A.~{Nehorai}, Beamforming using the
	fractional Fourier transform, IEEE Trans. Signal Process. \newblock  51~(6)
	(2003) 1663--1668.
	
	\bibitem {492554}B.~{Santhanam}, J.~H. {McClellan}, The discrete rotational
	{F}ourier transform, IEEE Trans. Signal Process. \newblock 44~(4) (1996) 
	994--998.
	
	\bibitem {S2011}E.~Sejdi\'c, I.~Djurovi\'c, L.~Stankovi\'c,
	\href{http://www.sciencedirect.com/science/article/pii/S0165168410003956}{Fractional
		Fourier transform as a signal processing tool: An overview of recent
		developments}, Signal Process. \newblock  91~(6) (2011) 1351 -- 1369.
		
	\bibitem{Shin}
	S. {Shinde}, V.~M. {Gadre}, An uncertainly principle for real signals in 
	the fractional {F}ourier transform domain, IEEE Trans. Signal Process. 
	\newblock  49~(11) (2001) 2545--2548.
		
	\bibitem{StR}
	E.~M. Stein, R. Shakarchi, {F}ourier analysis: an introduction, Princeton 
	Univ. Press, Princeton, N.J., 2003.
	
	\bibitem{StW}
	E.~M. Stein, G.~Weiss, Introduction to {F}ourier analysis on
	{E}uclidean spaces, Princeton Univ. Press, Princeton, N.J., 1971.

	
	\bibitem {T2010}R.~{Tao}, Y.~{Li}, Y.~{Wang}, Short-time fractional Fourier
	transform and its applications, IEEE Trans. Signal Process. \newblock  
	58~(5)
	(2010) 2568--2580.
	
	\bibitem {JntTranCor}J.~M. Vilardy, Y.~Torres, M.~S. Mill{\'{a}}n,
	E.~P{\'{e}}rez-Cabr{\'{e}},
	\href{https://doi.org/10.1088\%2F2040-8978\%2F16\%2F12\%2F125405}{Generalized
		formulation of an encryption system based on a joint transform 
		correlator and
		fractional {F}ourier transform}, J. Opt. \newblock  16~(12) (2014) 
		125405.
	
	\bibitem {zayed}A.~I. {Zayed}, {H}ilbert transform associated with the
	fractional {F}ourier transform, IEEE Signal Process. Lett. \newblock  5~(8)
	(1998) 206--208.
	
	
	\bibitem{zay982}
	A.~I. {Zayed}, A convolution and product theorem for the fractional 
	{F}ourier transform, IEEE Signal Process. Lett.  \newblock  5~(4) (1998) 
	101--103.

	
	\bibitem {chao}H.~Zhao, Z.~Zhong, W.~Fang, H.~Xie, Y.~Zhang, M.~Shan,
	Double-image encryption using chaotic maps and nonlinear non-{D}{C} joint
	fractional {F}ourier transform correlator, Opt. Eng. 55~(9) (2016) 093109.
\end{thebibliography}

\end{document}